%% file: exchange_QM.tex
\documentclass[pmlr,twocolumn,a4paper,10pt]{jmlr}
\usepackage{physics}
\usepackage{mathtools} 
  \mathtoolsset{mathic}
  \allowdisplaybreaks
\usepackage[T1]{fontenc}
\usepackage{textcomp}
\usepackage{lmodern}
\usepackage{mathptmx}
\usepackage[scaled]{helvet}
\usepackage{courier}
\usepackage{microtype}
\jmlrpages{}
\jmlrproceedings{}{} 
\jmlrvolume{}
\jmlryear{}

 \usepackage[british]{babel}
 \usepackage{mathrsfs}
\usepackage{booktabs} 
\usepackage{tikz} 
\include{macros}
\usepackage{tikz}
\usepackage{braids,calc}
  
\jmlrworkshop{}

\title[Quantum indistinguishability]{Quantum indistinguishability
through exchangeable desirable gambles
}

\author{
  \Name{Alessio Benavoli}\Email{alessio.benavoli@tcd.ie}\\ 
  \addr School of Computer Science and Statistics, Trinity College Dublin, Ireland
  \AND
  \Name{Alessandro Facchini}\Email{alessandro.facchini@idsia.ch}\\
    \Name{Marco Zaffalon}\Email{zaffalon@idsia.ch}\\
  \addr IDSIA (USI-SUPSI), Lugano, Switzerland
}

\begin{document}
\maketitle

\begin{abstract}
  Two particles are identical if all their intrinsic properties, such as spin and charge, are the same, meaning that no quantum experiment can distinguish them. In addition to the well known principles of quantum mechanics, understanding systems of identical particles requires a new postulate, the so called \emph{symmetrization postulate}. 
  In this work, we show that the postulate corresponds to exchangeability assessments for sets of  observables (gambles) in a quantum experiment, when quantum mechanics is seen as a normative and algorithmic theory guiding an agent to assess her subjective beliefs represented as (coherent) sets of gambles. Finally, we show how  sets of exchangeable observables (gambles) may be updated after a measurement and discuss the issue of defining entanglement for indistinguishable particle systems. 
\end{abstract}
\begin{keywords}
  quantum theory, indistinguishable particles, exchangeability, desirable gambles.
  \end{keywords}

\section{Introduction}
In recent works \citep{benavoli2016quantum,benavoli2017gleason} and in particular in  \citep{Benavoli2019bb},  we defined a theory of probability on a continuous space of complex vectors that complies with the two postulates of coherence (``The theory should be logically consistent''), and of computation (``Inferences in the theory should be computable in polynomial time''). We then showed that its deductive closure is 
tantamount to Quantum Mechanics (QM).  Hence QM may be viewed as a normative and algorithmic theory guiding an agent to assess her subjective beliefs represented as (coherent) sets of gambles 
on the results of a quantum experiment.
We were then able to derive (in a coherent way) the main postulates of QM from standard operations in probability theory (updating, marginalisation, time coherence). This means we derived a theory of probability which  theoretically and empirically agrees with QM experiments.

When one considers systems including  more than one particle, we must consider the implications of another important empirical observation: in
many of these systems the particles of interest belong to distinct classes of indistinguishable (identical) particles. Two (or more) particles are said to be identical if all their  properties (charge, mass, spin, etc.) 
are exactly the same. In other words, no experiment can distinguish one from the other.
Hence, all the electrons in the universe are identical, as are all the protons. 
This means that, when a physical
system contains two identical particles, there is no change in its properties if the roles of these two particles are exchanged.

This law is formulated in QM by the \textit{symmetrization postulate}, which establishes that in a system containing identical particles the only possible configurations of their
properties (e.g., spin) are either all symmetrical or all antisymmetrical with respect
to permutations of the labels of the particles. In the first case, the particles are called \textit{bosons}; in the second case they are called \textit{fermions}.

In this paper, we aim to derive the \textit{symmetrization postulate} from the way a subject accepts gambles on experiments involving indistinguishable particles.
We assume that the particles are \textit{exchangeable}, meaning roughly that the subject believes that the labels (i.e.\ electron 1, electron 2,..) we use to denote them, has no influence on the decisions and inferences she will make regarding the particles. 

Exchangeability is a fundamental concept in classical probability theory and statistics \citep{diaconis1980finite,regazzini1991coherence}. Its assumption, and the analysis of its consequences, goes back to \cite{finetti1974-5} and his famous \textit{Representation Theorem}. In statistics, this theorem is interpreted as stating that ``a sequence of random variables is exchangeable if it is conditionally independent and identically distributed.''
This theorem was generalised to QM by \cite{caves2002unknown} for \textit{quantum-state tomography}, which is a technique to estimate the density matrix of a particle by performing repeated measures (the order of the measures is assumed to be exchangeable).

In this paper, we instead deal with the exchangeability of indistinguishable particles. We show that we can derive the \textit{symmetrization postulate} by using the general framework for exchangeable gambles proposed by \citet{de2012exchangeability} for classical (imprecise) probability theory.\footnote{Exchangeability in the context of 
imprecise probability was originally proposed by \cite[Sec.\,9.5]{walley1991}}
This confirms, once again, that QM is a subjective theory of probability.

The rest of the paper is organised as follows. In Section \ref{sec:algorat} we recall how QM can be formulated, and thus understood, as an algorithmic theory of desirable gambles. After formulating in Section \ref{sec:symprinc} the symmetrisation postulate, in the following Section \ref{sec: exchange} we derive it in terms of (algorithmic) coherence and exchangeability. 
Finally, in Section \ref{sec:updat} we show how sets of exchangeable observables (gambles) may be updated after a measurement and in Section \ref{sec:ent} we discuss the issue of defining entanglement for indistinguishable particle systems

\section{Algorithmic rationality and QM}\label{sec:algorat}
 In this section, we recall some definition and results from  \citep{Benavoli2019bb}.
Consider a systems of $m$ particles (each one is an $n_j$-level system, for instance if we consider the spin of an electron $n_j=2$: the spin can be ``up'' or ``down''). When $m >1$ the system is said to be composite, whereas in case $m=1$ we are considering a single particle system. Hence, the  possibility space is
 $$
 \pspace=\times_{j=1}^m \overline{\complex}^{n_j}.
 $$
 where
 $$
\overline{\complex}^{n_j}=\{ x\in \complex^{n_j}: ~x^{\dagger}x=1\}.
$$
Next, we describe the observables, the gambles in our setting. Let us recall that in QM any real-valued observable is described by a Hermitian operator (matrix).
This naturally imposes restrictions on the type of `permitted gambles' $g$ on  a quantum experiment.
For a single particle, given a Hermitian operator $G\in \He^{n\times n}$
(with $\He^{n\times n}$ being the set of Hermitian matrices of dimension $n \times n$), a gamble on $x \in \overline{\complex}^{n}$  can be defined as:
$$
g(x)=x^\dagger G x.
$$
Since $G$ is Hermitian and  $x$ is bounded ($x^{\dagger}x=1$), $g$ is a real-valued bounded function. For a composite system of $m$ particles,  the gambles  are $m$-quadratic forms:
\begin{equation}
 \label{eq:gamble_many}
 g(x_1,\dots,x_m)=(\otimes_{j=1}^m x_j)^\dagger G (\otimes_{j=1}^m x_j),
\end{equation}
with $G \in \He^{n \times n}$, $n=\prod_{j=1}^m n_j$, and where $\otimes$ denotes the tensor product between vectors regarded as column matrices.\footnote{Why the tensor product? In classical probability, structural assessments of independence/dependence are expressed via expectations on factorised gambles $g(x_1,\dots,x_m)=\prod_{j=1}^m g_j(x_j)$.
This factorised gamble can equivalently be written as \eqref{eq:gamble_many}, see \citep{Benavoli2019bb} for more details.}
Therefore, we have that
\begin{align*}
 \mathcal{L}_R = &\{g
 \mid 
 ~ G \in \He^{n \times n}
 \}
\end{align*}
is the restricted set of `permitted gambles' in a quantum experiment. We can  also define the subset of nonnegative gambles $\nonnegative_R\coloneqq \{g \in \gambles_R \mid \min g\geq0 \}$ and the subset of negative gambles  $\negative_R\coloneqq \{g \in \gambles_R \mid \max g < 0 \}$.\footnote{Notice that, since $g$ is a polynomial and $\pspace$ is bounded, $\min g = \inf g$ and $\max g = \sup g$.}

Since  $ \mathcal{L}_R$ is a vector space including the constant gambles ($G=cI$ with $I$ identity matrix),\footnote{The  constant functions take the form $g(x_1,\dots,x_m)=c (\otimes_{j=1}^m x_j)^\dagger I (\otimes_{j=1}^m x_j)=c$.} we can use standard desirability to impose
rationality principles (coherence) in the way a subject should accept gambles. However, this would not lead to QM. Indeed, as discussed in the Introduction, QM follows by the two principles of coherence and of computation.\footnote{ 
QM is a theory of bounded (algorithmic) rationality 
\citep{zaffalon2019b,Benavoli2019b}.  Generalised types of coherence were described in some detail in \citep{quaeghebeur2015accept}.}

As shown by \cite{gurvits2003classical}, for $m>1$ the problem of deciding whether a gamble is nonnegative, that is whether it belongs to $\nonnegative_R$, is NP-hard, thus leading to a violation of the aforementioned computation principle.\footnote{The infimum coincides with the minimum because gambles are bounded polynomials.}  
To fulfil the computation requirement, we therefore need to
change the meaning of `being nonnegative'   by considering a subset $\bnonnegative \subsetneq \nonnegative$ for which the membership problem is in P.
This is done by considering the following new set of ``tautologies'':
\begin{align*}
\Sigma^{\geq}\coloneqq\{g \in  \mathcal{L}_R
\mid G\geq0\}.
\end{align*}
That is, a gamble is `nonnegative' whenever $G$ is PSD. Note that $\Sigma^{\geq}$ is the so-called cone of \emph{Hermitian sum-of-squares} polynomials.

 What described above is the essence of the algorithmic rationality behind QM. In other words, the corresponding algorithmic theory of desirable gambles 
 is based on the following redefinition of the tautologies:
\begin{itemize}
\item $\bnonnegative$ should always be desirable,
\end{itemize}
The rest of the theory follows exactly the footprints of the standard theory of (almost) desirability. In particular, the deductive closure for a finite\footnote{In case of arbitrary set of assessments, we simply ask in addition for $\bdomain$ to be topologically closed.}  set of assessment $\mathcal{G}$ is defined by:
\footnote{`$\posi(\mathcal{A})$' denotes the conic  hull of a set of gambles $\mathcal{A}$. It is defined as $\posi(\mathcal{A})=\{\sum_i \lambda_ig_i: \lambda_i\in \reals^{\geq}, g_i\in \mathcal{A}\}$.}
\begin{itemize}
\item $\bdomain\coloneqq \posi(\bnonnegative \cup \mathcal{G})$.
\end{itemize}
And finally the coherence postulate simply states that 
\begin{itemize}
\item A set $\bdomain$ of desirable gambles is said to be \emph{A-coherent}  if and only if
$-1 \notin \bdomain$,
\end{itemize}
where `$A$' stands for the the fact that 
the algorithmic bounds of the coherence problem for a finite set of assessments are established according to the choice of $\bnonnegative$.

\begin{remark}
 In classical coherence, the tautologies are the set of all nonnegative gambles $\mathcal{L}^{\geq}_R$. This is the only difference w.r.t.\ QM. The classical axioms of desirability are: (i) $\mathcal{L}^{\geq}_R$ should always be desirable;
 (ii) $\domain\coloneqq \posi(\mathcal{L}^{\geq}_R \cup \mathcal{G})$; (iii) $-1 \notin \domain$. However, evaluating if a gamble belongs to $\mathcal{L}^{\geq}_R$ is NP-hard as discussed previously. 
 \end{remark}
 
 \begin{remark} There are different notions of desirability (almost, strict, real \citep{walley1991}); here we use the term  desirability for  \text{almost desirability}. A-coherence is an instance of almost desirability.
\end{remark}

We can finally associate a `probabilistic' interpretation through the dual of an A-coherent set. 
Let us consider  the dual space $\gambles_R^*$ of all bounded linear functionals $L: \gambles_R \rightarrow \reals$. With the additional condition that linear functionals preserve the unitary gamble, the dual cone of an A-coherent $\bdomain\subset \gambles_R$  is given by

\begin{equation}
\label{eq:dualL}
\bdomain^\circ:=\left\{L \in \mathsf{S} \mid L(g)\geq0, ~\forall g \in \mathcal{G}\right\},
\end{equation}
where $\mathsf{S}=\{L \in \gambles_R^* \mid L(1)=1,~~L(h)\geq0 ~~\forall h \in \bnonnegative\}$ is the set of states.
It is not difficult to prove that $\bdomain^\circ$ can actually equivalently be defined as:
\begin{align}
\label{eq:credaldef}
\mathcal{M}&:=\{ \rho \in \mathscr{S} \mid Tr(G \rho) \geq  0,~ ~\forall g \in \mathcal{G}\},
\end{align} 
where  $\mathscr{S}=\{ \rho \in \He^{n\times n} \mid \rho\geq0,~~Tr(\rho)=1\}$ is the set of all density matrices
and gambles $g$ are defined as in \eqref{eq:gamble_many}
and are essentially specified by the Hermitian matrix $G$.  We also show that they are generalised moment\footnote{In classical probability, given a (real) variable $x$   and an expectation operator $E$, the n-th (non-central) moment of $x$ is defined as $m_n:=E[x^n]$ (we can also define multivariate moments, e.g., $E[x_1^nx_2^m]$). Given a  sequence of moments $m_0,m_1,m_2,\dots,m_n$, there exist infinitely many probability distributions corresponding to the same moments and they form a convex set. A sequence of scalars $m_0,m_1,m_2,\dots,m_n$ is  a valid sequence of moments  provided that they  satisfy certain consistency constraints. For instance, the moment matrix, obtained by organising that sequence into a matrix (in a certain way), must be positive semi-definite. This gives reason for the constraint $\rho\geq0$ for density matrices in QT. In general, $\rho$ is a generalised moment matrix, that is a moment matrix computed with respect to a `charge`. \citep{Benavoli2019bb}.} matrices: $\rho:=L(zz^\dagger)$. 

The derivation allows us to formulate quantum weirdness (that is the disagreement between QM and classical physics) as a Dutch book (sure loss).
This goes as follows. Given that QM uses a stronger notion of positivity/negativity, a set of desirable gambles can include  a gamble $f \in \mathcal{L}_R^{<}\backslash \Sigma^{<}$ and still be A-coherent.
When this happens, we have entanglement. In this case, the experimental results  appear illogical to us  (incompatible with our common understanding), because they are simply incoherent under  classical desirability.

\section{The symmetrisation postulate}\label{sec:symprinc}
 In this section, we formulate the symmetrisation postulate using QM theory \citep[XIV.C-1, p. 1434]{cohen2020quantum}. In the next section, we will instead derive this postulate using exchangeable gambles.
 
Suppose we have $m$ particles, each with single-particle  state space represented by a vector space $V=\overline{\mathbb{C}}^n$ (we assume $n_j=n$, same dimension for all particles). We denote a state (a wavefunction) with $\ket{\psi}$, where $\ket{\psi} \in V$.\footnote{$\ket{\psi}$ is a ket, that is a column vector.} According to QM postulates, if the particles were distinguishable the composite space of $m$ particles would be given by $\otimes_{i=1}^m V$.
Let us denote the state of a particle with $\ket{\alpha_i}$, so that an element of $\otimes_{i=1}^m V$ is denoted as 
$\ket{\psi}=\ket{\alpha_1} \otimes \dots \otimes \ket{\alpha_m}$.

\begin{remark}
 In section \ref{sec:algorat} we considered $x_i \in V$, while in this section we use $\ket{\alpha_i} \in V$. Why? The reason is that, in Section \ref{sec:algorat}, $x_i$ represents an unknown ``classical'' variable (e.g., the direction of the spin) and we ask a subject to express her beliefs about $x_i$ in terms of acceptance of gambles. Conversely, $\ket{\alpha_i} $ is a state: a proxy quantity which is used in QM to compute the probability of the results of an experiment. QM postulates are formulated in terms of $\ket{\alpha_i} $ (usually denoted as $\ket{\psi_i}$). Indeed, under the epistemic interpretation of QM, $\ket{\alpha_i} $ corresponds to a belief state  and so it is different from  $x_i$. This difference is also evident from the fact that, for a composite system, $\ket{\psi}=\ket{\alpha_1} \otimes \dots \otimes \ket{\alpha_m} \in \otimes_{i=1}^m V$, while $[x_1,\dots,x_m]\in \times_{i=1}^m V$. To understand this difference, consider the toss of a classical coin: $\Omega=\{H,T\}$ and  $p=[p_H,p_T] \in \mathbb{R}^2$ is the vector of probabilities for Heads and Tails. Now consider the toss of three coins, the composite possibility space 
is $ \times_{i=1}^3 \Omega$, while the joint probability mass function belongs to $\otimes_{i=1}^3\mathbb{R}^2=\mathbb{R}^8$.
 \end{remark}

In this work, we are  interested in defining the state space for indistinguishable particles. 

Let $\pi$ 
denotes a permutation of the indices of the elements of the tensor product $\ket{\alpha_1} \otimes \dots \otimes \ket{\alpha_m}$. 
Since such a permutation defines the product $\ket{\alpha_{\pi(1)}} \otimes \dots \otimes \ket{\alpha_{\pi(m)}}$, by permuting the elements of the tensor products, we are basically permuting the labels of the particles.
A permutation that only swaps two variables is called a \textit{transposition}.

The \textit{sign of a permutation} $\pi$, denoted by $\sign(\pi)$,  equals
1   if  $\pi$ can be written as  a product
of an even number of transpositions, and equals -1 if $\pi$ can be written as a product of an odd number of transpositions.
Notice that the sign of $\pi$ can be calculated as follows:
$$
\sign(\pi)=\text{det} \sum_{i=1}^m e_i e^T_{\pi(i)},
$$   
where $ e_i$ is an element of the canonical basis of $\mathbb{R}^m$ (see \citep[XIV.B-2-c]{cohen2020quantum}).

Since permutations are linear operator, we can equivalently express permutation $\pi$ as a matrix operator $P_{\pi}$ acting on the tensor product:
$$
P_{\pi} (\ket{\alpha_1} \otimes \dots \otimes \ket{\alpha_m}):=\ket{\alpha_{\pi(1)}} \otimes \dots \otimes \ket{\alpha_{\pi(m)}}.
$$
The matrix $P_{\pi}$ is unitary, that is 
$P^{\dagger}_{\pi}P_{\pi}=P_{\pi}P^{\dagger}_{\pi}=I$, but not necessarily Hermitian. In what follows, by $\mathbb{P}_m$ we both denote the collection of all permutations and of all corresponding permutation operators.

We now introduce the  \textit{symmetriser} and the \textit{antisymmetriser}:
\begin{align*}
\Pi_\mathrm{Sym}&:= \frac{1}{m!} \sum_{\pi_r \in \mathbb{P}_m} P_{\pi_r}, \\
\Pi_\mathrm{Anti}&:= \frac{1}{m!} \sum_{\pi_r \in \mathbb{P}_m} \sign(\pi_r) P_{\pi_r}.
\end{align*}
which  are projectors\footnote{They are Hermitian $\Pi_\mathrm{Sym}^{\dagger}=\Pi_\mathrm{Sym},\Pi_\mathrm{Anti}^{\dagger}=\Pi_\mathrm{Anti}$ and they satisfy $\Pi_\mathrm{Sym}^2=\Pi_\mathrm{Sym},\Pi_\mathrm{Anti}^2=\Pi_\mathrm{Anti}$ and $\Pi_\mathrm{Sym}\Pi_\mathrm{Anti}=\Pi_\mathrm{Anti}\Pi_\mathrm{Sym}=0$.} \citep[XIV.B-2-c]{cohen2020quantum}. They project onto respectively:
\begin{align*}
 \mathrm{Sym}^m V&=\{\ket{\psi} \in \otimes_{i=1}^m V: P_{\pi}\ket{\psi}= \ket{\psi},~~\forall \pi \in \mathbb{P}_m\}\\
  \mathrm{Anti}^mV&=\{\ket{\psi} \in \otimes_{i=1}^m V: P_{\pi}\ket{\psi}= \sign(\pi)\ket{\psi},~~ \forall \pi \in \mathbb{P}_m\}.
\end{align*}

\begin{lemma}[\cite{cohen2020quantum}]
\label{eq:lem1}
The following equalities hold for any permutation operator $P_{\pi}  \in \mathbb{P}_m$:
 \begin{enumerate}
  \item $P_{\pi} \Pi_\mathrm{Sym} = \Pi_\mathrm{Sym} P_{\pi} =\Pi_\mathrm{Sym}$;
    \item $P_{\pi} \Pi_\mathrm{Anti} = \Pi_\mathrm{Anti} P_{\pi} =\sign(\pi)\Pi_\mathrm{Anti}$.
 \end{enumerate}
\end{lemma}
\begin{proof}
 Given two permutations $P_{\pi_i} \neq P_{\pi_j}$, we have that $P_{\pi}P_{\pi_i} \neq P_{\pi} P_{\pi_j}$.
 Hence we have that
 \begin{align*}
 P_{\pi} \Pi_\mathrm{Sym}=\frac{1}{m!} \sum_{\pi_r \in \mathbb{P}_m} P_{\pi}P_{\pi_r}=\frac{1}{m!} \sum_{\pi'_r \in \mathbb{P}_m} P_{\pi'_r}. 
 \end{align*}
Analogously, since  $\sign(\pi)\sign(\pi)=1$
 \begin{align*}
 P_{\pi} \Pi_\mathrm{Anti}&=\frac{1}{m!} \sum_{\pi_r \in \mathbb{P}_m} \sign(\pi_r) P_{\pi}P_{\pi_r}\\
 &=\frac{\sign(\pi)}{m!} \sum_{\pi_r \in \mathbb{P}_m} \sign(\pi_r) \sign(\pi)P_{\pi}P_{\pi_r}\\
 &=\frac{\sign(\pi)}{m!} \sum_{\pi'_r \in \mathbb{P}_m} \sign(\pi'_r) P_{\pi'_r}. 
 \end{align*}
\end{proof}

The \emph{symmetrisation postulate} states the following:
\begin{quote}
 When a system includes several identical particles, only certain states of its state space can describe its physical states.  Physical states are, depending on the nature of the identical particles, either completely symmetric or completely antisymmetric with respect to permutation of these particles. Those particles for which the physical states are symmetric are called bosons, and those for which they are antisymmetric, fermions.  \citep[XIV.C-1, p. 1434]{cohen2020quantum}
\end{quote}
The  postulate thus limits the state space (possibility space) for a system of identical particles. Contrary to the case of particles of different natures, this space is no longer the tensor product $\otimes_{i=1}^m V$ of the individual state spaces of the particles constituting the system, but rather a subspace, namely $\mathrm{Sym}^m V$ or $\mathrm{Anti}^m V$, depending on whether the particles are bosons or fermions. Only states belonging either to $\mathrm{Sym}^m V$ or to $\mathrm{Anti}^m V$ are physically possible.  This is the reasons they are called \textit{physical states}.

Given $k$ physical states $\ket{\psi_i}$ (belonging to either $\mathrm{Sym}^m V$ or $\mathrm{Anti}^m V$), we can then define the density matrix as usual:
$$
\rho=\sum_{i=1}^{k} p_i  \ket{\psi_i}\bra{\psi_i}, 
$$
where $p_i$ are probabilities, $p_i\geq0$ and $\sum_{i=1}^kp_i=1$.
It can then be verified that, in the symmetric case, given that $\ket{\psi_i}=\Pi_\mathrm{Sym}\ket{\psi_i}$, we have that $\rho=\Pi_\mathrm{Sym}\rho\Pi_\mathrm{Sym}$.
Similarly, in the antisymmetric case,
$\rho=\Pi_\mathrm{Anti}\rho\Pi_\mathrm{Anti}$.

\begin{example}
\label{eq:1}
 Consider $m=2$ particles with $\ket{\alpha_1},\ket{\alpha_2} \in \overline{\mathbb{C}}^2$. In this case there are  only two possible
 permutations $\pi_a$ (identity) and $\pi_b$ (swap)
 with $\sign(\pi_b)=-1$: 
 \begin{center}
\begin{tikzpicture}[
  /pgf/braid/.cd,
  style strands={1}{red},
  style strands={2}{blue},
  number of strands=2
  ]
\braid (identity) at (0,0) 1 ;
\node at ([yshift=1cm]identity) {$1~~~~~~~~~2$};
\node at ([yshift=-1cm]identity) {$1~~~~~~~~~2$};

\braid (12) at ([xshift=1.3cm]identity-2-s) s_1;
\node at ([yshift=1cm]12) {$1~~~~~~~~~2$};
\node at ([yshift=-1cm]12) {$2~~~~~~~~~1$};
\end{tikzpicture}
\end{center}
The permutation matrices are $P_{\pi_a}=I$ and:
\begin{align}
\label{eq:Permexb}
P_{\pi_b}=\begin{bmatrix}
 1 &  0 & 0 & 0\\
 0 & 0 & 1 &0\\
0 & 1 & 0 &0\\
0 & 0 & 0 & 1
\end{bmatrix}.
\end{align}
The latter acts on $\ket{\alpha_1} \otimes \ket{\alpha_2}$ as follows
$$
P_{\pi_b}(\ket{\alpha_1} \otimes \ket{\alpha_2})=P_{\pi_b}\begin{bmatrix}
 \alpha_{11}\alpha_{21}\\
 \alpha_{11}\alpha_{22}\\
 \alpha_{12}\alpha_{21}\\
 \alpha_{12}\alpha_{22}
\end{bmatrix}=\begin{bmatrix}
 \alpha_{11}\alpha_{21}\\
  \alpha_{12}\alpha_{21}\\
 \alpha_{11}\alpha_{22}\\
 \alpha_{12}\alpha_{22}
\end{bmatrix}=\ket{\alpha_2} \otimes \ket{\alpha_1}.
$$
The projectors are:
\begin{equation}
\label{eq:Proj}
\Pi_\mathrm{Sym}=\frac{I+P_{\pi_b}}{2},~~\Pi_\mathrm{Anti}=\frac{I+\text{sign}(\pi_b)P_{\pi_b}}{2}=\frac{I-P_{\pi_b}}{2},
\end{equation}
which act on $\ket{\alpha_1} \otimes \ket{\alpha_2}$ as follows\footnote{The right hand side term in \eqref{eq:Psym1} or \eqref{eq:Panti1} is a complex vector, but its norm can be different from one. In this latter case, it needs to be normalised.}

\begin{align}
\label{eq:Psym1}
\Pi_\mathrm{Sym}(\ket{\alpha_1} \otimes \ket{\alpha_2})&=\begin{bmatrix}
 \alpha_{11}\alpha_{21}\\
 \tfrac{\alpha_{11}\alpha_{22}+ \alpha_{12}\alpha_{21}}{2}\\
 \tfrac{\alpha_{11}\alpha_{22}+ \alpha_{12}\alpha_{21}}{2}\\
 \alpha_{12}\alpha_{22}
\end{bmatrix},\\
\label{eq:Panti1}
\Pi_\mathrm{Anti}(\ket{\alpha_1} \otimes \ket{\alpha_2})&=\begin{bmatrix}
 0\\
 \tfrac{\alpha_{11}\alpha_{22}- \alpha_{12}\alpha_{21}}{2}\\
 \tfrac{\alpha_{12}\alpha_{21}-\alpha_{11}\alpha_{22}}{2}\\
 0\\
\end{bmatrix}
\end{align}
From last equality, it follows that, in case $\alpha_1=\alpha_2$, $\Pi_\mathrm{Anti}(\ket{\alpha_1} \otimes \ket{\alpha_2})=0$.
This is called \emph{Pauli exclusion principle}: two Fermions cannot have identical state.
\end{example}

\section{Exchangeable gambles}\label{sec: exchange}
In the previous section, we discussed the symmetrisation postulate. In this section, we formulate it in terms of A-coherence and exchangeability. In doing so, we extend some of the definitions and results originally presented in \citep{de2012exchangeability} to the quantum setting introduced in Section \ref{sec:algorat}.   

As discussed in Section \ref{sec:algorat}, we consider gambles on $x_i \in V=\overline{\mathbb{C}}^n$. 
Given $m$ particles, the possibility space 
is $\times_{i=1}^m V$. Therefore, $\pi$ 
denotes a permutation of the indices of the vector $(x_1,\dots,x_m)$, i.e.,
$$
\pi(x_1,\dots,x_m)=(x_{\pi(1)},\dots,x_{\pi(m)}).
$$
A generic gamble is denoted as:
\begin{align*}
 g(z,z):=z^\dagger G z,
\end{align*}
with
 $z:=\otimes_{j=1}^m x_j$. Let $\pi_r,\pi_l$ be two permutations, 
  we define
\begin{align*}
\pi_lg(z,z)\pi_r&:= \frac{1}{2}\left(g(\pi_l z,\pi_r z)+g(\pi_r z,\pi_l z)\right)\\
&=\frac{1}{2}\left(z^\dagger P^{\dagger}_{\pi_l}G P_{\pi_r}z+z^\dagger P^{\dagger}_{\pi_r}G P_{\pi_l}z\right).
\end{align*}
Note that (i) $\pi_lg\pi_r=\pi_rg\pi_l$, and
(ii) $\pi_lg\pi_r$ is a gamble (it returns real values).\footnote{This holds because $P^{\dagger}_{\pi_l}G P_{\pi_r}+ P^{\dagger}_{\pi_r}G P_{\pi_l}$ is Hermitian.}
\begin{remark}
  This definition of permuted gamble is different from the one used in \citep{de2012exchangeability} (the permutation of $g(\omega)$ is defined as $\pi \circ g=g(\pi \omega)$). In QM, gambles are quadratic forms of complex variables and, therefore, we can define more general symmetries by exploiting the fact that $z$ and its complex conjugate $z^{\dagger}$ can be treated as two ``different'' variables.
\end{remark}

\begin{example}
 Consider $m=2$ particles with $x_1,x_2 \in \overline{\mathbb{C}}^2$. We have already seen that there are  only two possible
 permutations $\pi_a$ (identity) and $\pi_b$ (swap). Therefore, we have $\pi_ag\pi_a=g$ and
\begin{align*}
\nonumber
\pi_ag\pi_b&= \frac{1}{2}\left(
(x_1\otimes x_2)G(x_2\otimes x_1)+(x_2\otimes x_1)G(x_1\otimes x_2)\right),\\
\nonumber
\pi_bg\pi_b&= (x_2\otimes x_1)G(x_2\otimes x_1).
\end{align*}
\end{example}

For $\pi_l,\pi_r \in \mathbb{P}_m$,  we write
  \[
\delta^\star_{l,r} := 
\begin{cases}
\sign(\pi_l)\sign(\pi_r)  & \text{ when }\star = {Anti},\\
1 & \text{ when }\star = {Sym}.
\end{cases}
\]

Given this definition, in the remaining of this section, all definitions, results and corresponding proofs will be parameterised by $\star \in \{{Anti}, {Sym}\}$ and $\delta^\star_{l,r}$. They therefore apply, uniformly, to both the symmetric and the antisymmetric  cases. 
 
 We now provide the definition of A-coherent $\star$-exchangeable set of desirable gambles.
 \begin{definition}
\label{def:1}
Consider the set
\begin{align*}
\mathcal{I}_{\star}&:= \{g- \delta^\star_{l,r}\pi_lg\pi_r \mid g \in \gambles_R, ~\pi_l,\pi_r \in \mathbb{P}_m\}.
\end{align*}
We say that an A-coherent set of desirable gambles $\bdomain$ is \emph{$\star$-exchangeable} if  $\mathcal{I}_{\star} \subseteq \bdomain$.
 \end{definition}
 
 Given Definition \ref{def:1}, we can prove the following result.
 
\begin{proposition}
\label{prop:1}
Let $\bdomain$ be an A-coherent set of desirable gambles.
 If $\bdomain$ is $\star$-exchangeable, then it is also $\star$-permutable,  that is $\delta^\star_{l,r} \pi_lg\pi_r$   are in $\bdomain$ for all $g \in \bdomain$ and all $\pi_l,\pi_r \in \mathbb{P}_m$.  
\end{proposition}
\begin{proof} The proof is similar as the one for
 \citep[Prop.9]{de2012exchangeability}.  For $g \in \bdomain$ and $\pi_l,\pi_r \in \mathbb{P}_m$, we have $-g- \delta^\star_{l,r}\pi_l(-g)\pi_r \in  \mathcal{I}_{\star} \subseteq \bdomain$. Given that $-g=z^\dagger  (-G)  z$, then $-g-\delta^\star_{l,r} \pi_l (-g)\pi_r=\delta^\star_{l,r} \pi_lg\pi_r-g$.
Since
$\delta^\star_{l,r} \pi_lg\pi_r=\delta^\star_{l,r} \pi_lg\pi_r-g+g$ and $g,\delta^\star_{l,r} \pi_lg\pi_r-g \in \bdomain$, we conclude by additivity that $\delta^\star_{l,r}  \pi_lg\pi_r \in \bdomain$. 
\end{proof}

As in \citep{de2012exchangeability}, but taking into account that we are working with quadratic forms, we define the linear operators
\begin{align*}
 \text{ex}_{\star}^m(g)&:=z^{\dagger} \Pi_{\star}^{\dagger} G \Pi_{\star} z.
 \end{align*}

We verify some of their properties; in particular  that  they can be used to equivalently characterise symmetric and antisymmetric exchangeability (Corollary \ref{cor:exc}). 

The first result follows immediately from the fact that the symmetrisers and the antisymmetriser are projectors.
\begin{lemma}\label{lem:projo}
Let $g$ be a gamble, then 
$\text{ex}_{\star}^m(\text{ex}_{\star}^m(g)) = \text{ex}_{\star}^m(g)$.
\end{lemma}
 The idea behind this linear transformations $\text{ex}_{\star}^m(g)$ is that they render a gamble $g$ insensitive to permutation by replacing it with the uniform average $\text{ex}_{\star}^m(g)$ of all its permutations $\pi_l g\pi_r$, as shown hereafter.
 \begin{proposition}\label{prop:sum}
 Let $g$ be a gamble, then 
$$
\text{ex}_{\star}^m(g) =   \frac{1}{m!m!} \sum_{\pi_r,\pi_l  \in \mathbb{P}_m} \delta^\star_{l,r} \pi_lg\pi_r.\,
  $$
 \end{proposition}
 \begin{proof}
 It is immediate to verify that $\text{ex}_{\star}^m(g) =   \frac{1}{m!m!} \sum_{\pi_r,\pi_l  \in \mathbb{P}_m} \delta^\star_{l,r} g(\pi_l z,\pi_r z)$. To conclude, note that: 
  $$
  \begin{aligned}
&\sum_{\pi_l,\pi_r \in \mathbb{P}_m} \delta^\star_{l,r} \pi_lg\pi_r = \\
& =\sum_{\pi_l,\pi_r \in \mathbb{P}_m} \frac{\delta_\star}{2}\left(g(\pi_l z,\pi_r z)+g(\pi_r z,\pi_l z)\right)\\
 &=\frac{1}{2}\sum_{\pi_l,\pi_r \in \mathbb{P}_m} \delta^\star_{l,r} g(\pi_l z,\pi_r z)+\frac{1}{2}\sum_{\pi_l,\pi_r \in \mathbb{P}_m} \delta^\star_{r,l} g(\pi_r z,\pi_l z)\\
 &=\sum_{\pi_l,\pi_r \in \mathbb{P}_m} \delta^\star_{l,r} g(\pi_l z,\pi_r z).
  \end{aligned}
$$ 
 \end{proof}
 Clearly, the linear transformations $\text{ex}_{\star}^m$ assume the same value for all gambles that can be related to each other
through some permutation.
 \begin{proposition}\label{prop:star}
Let $g$ be a gamble, and $\pi_l,\pi_r \in \mathbb{P}_m$. Then 
$$
  \text{ex}_{\star}^m(\delta^\star_{l,r}\pi_lg\pi_r)=\text{ex}_{\star}^m(g).
  $$
 \end{proposition}
 \begin{proof} By exploiting linearity 
 \begin{align*}
 &\text{ex}_{\star}^m(\delta^\star_{l,r} \pi_lg\pi_r)= \delta^\star_{l,r}\text{ex}_{\star}^m( \pi_lg\pi_r) =\\
&= \delta^\star_{l,r} \big( z^{\dagger}\Pi_{\star}^{\dagger}\left(\frac{1}{2}\left(P^{\dagger}_{\pi_l}G P_{\pi_r}+ P^{\dagger}_{\pi_r}G P_{\pi_l}\right)\right)\Pi_{\star}z \big) \\
&= 
\frac{\delta^\star_{l,r}}{2}  z^{\dagger}\Pi_{\star}^{\dagger}P^{\dagger}_{\pi_l}G P_{\pi_r}\Pi_{\star}z + \frac{\delta^\star_{l,r}}{2} z^{\dagger}\Pi_{\star}^{\dagger}P^{\dagger}_{\pi_r}G P_{\pi_l}\Pi_{\star}z \\
&= 
\frac{\delta^\star_{l,r}}{2}  z^{\dagger}(P_{\pi_l}\Pi_{\star})^{\dagger}G (P_{\pi_r}\Pi_{\star})z + \frac{\delta^\star_{l,r}}{2} z^{\dagger}(P_{\pi_r}\Pi_{\star})^{\dagger}G (P_{\pi_l}\Pi_{\star})z
\end{align*}
 By Lemma \ref{eq:lem1} and the fact that $\delta^\star_{l,r}\delta^\star_{l,r}=1$, we finally obtain
 \begin{align*}
& \frac{\delta^\star_{l,r}}{2}  z^{\dagger}(P_{\pi_l}\Pi_{\star})^{\dagger}G (P_{\pi_r}\Pi_{\star})z + \frac{\delta^\star_{l,r}}{2} z^{\dagger}(P_{\pi_r}\Pi_{\star})^{\dagger}G (P_{\pi_l}\Pi_{\star})z= \\
& =\frac{\delta^\star_{l,r}\delta^\star_{l,r}}{2}  z^{\dagger}\Pi_{\star}^{\dagger}G\Pi_{\star}z + \frac{\delta^\star_{l,r}\delta^\star_{l,r}}{2} z^{\dagger}\Pi_{\star}^{\dagger}G \Pi_{\star}z\\
& = \text{ex}_{\star}^m(g).
 \end{align*}
 \end{proof}

Similarly to what was done by  \cite{de2012exchangeability}, we can prove the following.

\begin{corollary}
\label{cor:exc}
 Let $\bdomain$ be an A-coherent set of desirable gambles. Given 
\begin{align*}
\mathcal{V}_{\star}&:= \{g-  \text{ex}_{\star}^m(g) \mid g \in \gambles_R\}
\end{align*}
the following claims are equivalent, 
 \begin{description}
  \item[(1)]  $\bdomain$ is $\star$-exchangeable;
   \item[(2)]   
   $\mathcal{V}_{\star} \subseteq \bdomain$.
   \end{description}
 \end{corollary}
 \begin{proof}
For (1 $\Rightarrow$ 2), 
by Proposition \ref{prop:sum}, we can write  $g-  \text{ex}_{\star}^m(g)=\frac{1}{m!m!}\sum_{\pi_l\pi_r} ( g -\delta^\star_{l,r}\pi_lg\pi_r)$.
  Since $\bdomain$ satisfies linearity and given
  $\mathcal{I}_{\star} \subseteq \bdomain$, then
  $g-  \text{ex}_{\star}^m(g) \in \bdomain$.
  
  For (2 $\Rightarrow$ 1), 
 by linearity of $\text{ex}_{\star}^m$ and Proposition \ref{prop:star}
  $$
  g-\delta^\star_{l,r}\pi_lg\pi_r -  \text{ex}_{\star}^m(g-\delta^\star_{l,r}\pi_lg\pi_r)=  g-\delta^\star_{l,r}\pi_lg\pi_r,
  $$
  which shows that $  g-\delta^\star_{l,r}\pi_lg\pi_r \in \bdomain$.
 \end{proof}

The following result also holds.
\begin{proposition}
\label{prop:2}
Let $\bdomain$ be an A-coherent  set of desirable gambles. 
Then, assuming $\bdomain$ is $\star$-exchangeable, the following claims hold for all  gambles $g,g'$:
\begin{enumerate}
 \item $g \in \bdomain$ iff  $\text{ex}_{\star}^m(g) \in \bdomain$;
  \item if $\text{ex}_{\star}^m(g)=\text{ex}_{\star}^m(g')$ then $g \in \bdomain$ iff  $g' \in \bdomain$.
\end{enumerate}
\end{proposition}
\begin{proof} The proof is the same as for  \citep[Prop.10]{de2012exchangeability}. First notice that the first claim follows from the second, by taking $g':=\text{ex}_{\star}^m(g)$ and applying Lemma \ref{lem:projo}. For the second claim, assume $\text{ex}_{\star}^m(g)=\text{ex}_{\star}^m(g')$ and $g \in \bdomain$. Notice that $g'-  \text{ex}_{\star}^m(g')=g'-  \text{ex}_{\star}^m(g), -g-  \text{ex}_{\star}^m(-g)=   \text{ex}_{\star}^m(g) -g  \in \mathcal{V}_{\star}$. By Corollary \ref{cor:exc} and additivity, we obtain  $(g'-  \text{ex}_{\star}^m(g)) + (\text{ex}_{\star}^m(g) -g) + g = g' \in \bdomain$.
\end{proof}

We now consider the dual of an A-coherent set of 
$\star$-exchangeable gambles.

From Section \ref{sec:algorat}, to define the dual, we consider  the dual space $\gambles_R^*$ of all bounded linear functionals $L: \gambles_R \rightarrow \reals$. With the additional condition that linear functionals preserve the unitary gamble, the dual cone of an A-coherent $\bdomain\subset \gambles_R$  is given by
\begin{equation}
\label{eq:dualL1}
\bdomain^\circ=\left\{L \in \mathsf{S} \mid L(g)\geq0, ~\forall g \in \mathcal{G}\right\},
\end{equation}
where $\mathsf{S}=\{L \in \gambles_R^* \mid L(1)=1,~~L(h)\geq0 ~~\forall h \in \bnonnegative\}$ is the set of states.  

 \begin{definition}
Let $L \in \mathsf{S}$. 
We say that $L$ is \emph{$\star$-exchangeable} if  it belongs to the dual $\bdomain^\circ$ of an A-coherent  $\star$-exchangeable set of gambles $\bdomain$. 
 \end{definition}

\begin{proposition}
 Assume $L \in \mathsf{S}$. The following statements are equivalent: 
\begin{enumerate}
 \item $L$ is $\star$-exchangeable;
  \item  $L(f)=0$ for all $f \in \mathcal{I}_{\star}$.
 \item  $L(f)=0$ for all $f \in \mathcal{V}_{\star}$.
\end{enumerate}
\end{proposition}
\begin{proof}
We verify (1$\Leftrightarrow$2). If $L$ is  $\star$-exchangeable,  we know that $g- \delta^\star_{l,r} \pi_lg\pi_r, \delta^\star_{l,r} \pi_lg\pi_r-g\in \bdomain$, meaning that
$L(g- \delta^\star_{l,r} \pi_lg\pi_r)\geq 0$ and $-L(g- \delta^\star_{l,r} \pi_lg\pi_r)\geq 0$. Therefore $L(f)=L(g- \delta^\star_{l,r} \pi_lg\pi_r)=0$. 
For the other direction, assume that $L(f)=0$ for all $f \in \mathcal{I}_{\star}$, From $L$, by duality, we can define the set of desirable gambles $\{g \in \gambles_R: L(g)\geq0\}$. We have proven in \citep{Benavoli2019bb} that this is an A-coherent set of desirable gamble and, moreover, it includes 
$\mathcal{I}_{\star}$ by hypothesis. 
By Corollary \ref{cor:exc}, the equivalence (1$\Leftrightarrow$3) can be proven in a similar way.
\end{proof}
We recall the following well-know result (see e.g. \citep{holevo2011probabilistic}).
\begin{proposition}
\label{lem:tracepos}
Let $G$ be a Hermitian matrix; then $G\geq0$ if and only if
$Tr(SG)\geq0$ for all $S\geq0$.
\end{proposition}
We use the previous result to prove the following.
\begin{proposition}
\label{prop:exchangeL}
 Assume $L\in \mathsf{S}$. The following statements are equivalent: 
 \begin{enumerate}
  \item  $L$ is $\star$-exchangeable;
  \item $L\left(zz^\dagger-\frac{\delta^\star_{l,r}}{2}P_{\pi_r}zz^\dagger P^{\dagger}_{\pi_l}- \frac{\delta^\star_{l,r}}{2}P_{\pi_l}zz^\dagger P^{\dagger}_{\pi_r}\right)=0$ for all $\pi_l\pi_r \in \mathbb{P}_m$;
  \item $L\left(zz^\dagger-\Pi_{\star} zz^\dagger \Pi^\dagger_{\star}\right)=0$.
 \end{enumerate}
\end{proposition}
\begin{proof}
As before, we only verify the equivalence (1$\Leftrightarrow$2). 
Assume $L\in \mathsf{S}$ is $\star$-exchangeable and consider the set of gambles $\mathcal{A}=\{g-\delta^\star_{l,r}\pi_lg\pi_r: \pi_l\pi_r \in \mathbb{P}_m, G\geq0\}$ and $\mathcal{B}=\{\delta^\star_{l,r}\pi_lg\pi_r-g: \pi_l\pi_r \in \mathbb{P}_m,~G\geq0\}$.
 Since $L$ is $\star$-exchangeable, it follows that $L(f),L(f')\geq0$ for each $f \in \mathcal{A},f' \in \mathcal{B}$. This implies that
 $$
 \begin{aligned}
0&\leq L(g- \delta^\star_{l,r} \pi_lg\pi_r)\\
&=L(z^\dagger G z)-\frac{\delta^\star_{l,r}}{2}L\left(z^\dagger P^{\dagger}_{\pi_l}G P_{\pi_r}z+z^\dagger P^{\dagger}_{\pi_r}G P_{\pi_l}z\right)\\
 &=Tr\left(GL\left(zz^\dagger-\frac{\delta^\star_{l,r}}{2}P_{\pi_r}zz^\dagger P^{\dagger}_{\pi_l}- \frac{\delta^\star_{l,r}}{2}P_{\pi_l}zz^\dagger P^{\dagger}_{\pi_r}\right)\right)\\
  &=Tr\left(G\left(L(zz^\dagger)-\frac{\delta^\star_{l,r}}{2}P_{\pi_r}L(zz^\dagger) P^{\dagger}_{\pi_l}- \frac{\delta^\star_{l,r}}{2}P_{\pi_l}L(zz^\dagger) P^{\dagger}_{\pi_r}\right)\right)\\
\end{aligned}
 $$
 for each $\pi_l\pi_r \in \mathbb{P}_m,~G\geq0$.
 Similarly, we have that $0\leq L(-g+\delta^\star_{l,r}\pi_lg\pi_r)=-L(g- \delta^\star_{l,r} \pi_lg\pi_r)$.
We therefore conclude the proof of this implication by applying Proposition \ref{lem:tracepos}. 
 To prove the other direction, simply note that the second claim implies that $0= L(-g+ \delta^\star_{l,r}\pi_lg\pi_r)=-L(g-\delta^\star_{l,r}\pi_lg\pi_r)$.
\end{proof}

From \citep{Benavoli2019bb}, we know that $\rho:=L(zz^\dagger)$ is indeed a density matrix. Therefore,  Proposition \ref{prop:exchangeL} immediately implies the following.

\begin{corollary}\label{cor:symmprinc}
\label{co:exchangeRho} A density matrix $\rho \in \mathscr{S}=\{ \rho \in \He^{n\times n} \mid \rho\geq0,~~Tr(\rho)=1\}$ is  
$\star$- exchangeable if any of the following statement holds:
\begin{enumerate}
 \item $ \rho=\frac{\delta^\star_{l,r}}{2}P_{\pi_r} \rho P^{\dagger}_{\pi_l}+ \frac{\delta^\star_{l,r}}{2}P_{\pi_l} \rho P^{\dagger}_{\pi_r}$ for all $\pi_l\pi_r \in \mathbb{P}$;
  \item $ \rho=\Pi_{\star}\rho\Pi^{\dagger}_{\star}$.
\end{enumerate}
\end{corollary}
Point 2 of Corollary \ref{cor:symmprinc} therefore derives the \textit{symmetrisation postulate} discussed in Section \ref{sec:symprinc} via duality from a set of A-coherent exchangeable gambles.

\begin{example}
\label{eq:3}
 Consider the density matrix
 \begin{align}
\nonumber
\rho&:=L\left(\left[\begin{smallmatrix}x_{11} x_{11}^{\dagger} x_{21} x_{21}^{\dagger} & x_{11}^{\dagger} x_{12} x_{21} x_{21}^{\dagger} & x_{11} x_{11}^{\dagger} x_{21}^{\dagger} x_{22} & x_{11}^{\dagger} x_{12} x_{21}^{\dagger} x_{22}\\x_{11} x_{12}^{\dagger} x_{21} x_{21}^{\dagger} & x_{12} x_{12}^{\dagger} x_{21} x_{21}^{\dagger} & x_{11} x_{12}^{\dagger} x_{21}^{\dagger} x_{22} & x_{12} x_{12}^{\dagger} x_{21}^{\dagger} x_{22}\\x_{11} x_{11}^{\dagger} x_{21} x_{22}^{\dagger} & x_{11}^{\dagger} x_{12} x_{21} x_{22}^{\dagger} & x_{11} x_{11}^{\dagger} x_{22} x_{22}^{\dagger} & x_{11}^{\dagger} x_{12} x_{22} x_{22}^{\dagger}\\x_{11} x_{12}^{\dagger} x_{21} x_{22}^{\dagger} & x_{12} x_{12}^{\dagger} x_{21} x_{22}^{\dagger} & x_{11} x_{12}^{\dagger} x_{22} x_{22}^{\dagger} & x_{12} x_{12}^{\dagger} x_{22} x_{22}^{\dagger}\end{smallmatrix}\right]\right)\\
\label{eq:densentbos}
&=\frac{1}{2}\begin{bmatrix}
 1 & 0 & 0 & 1\\
  0 & 0 & 0 & 0\\
   0 & 0 & 0 & 0\\
    1 & 0 & 0 & 1\\
\end{bmatrix}
\end{align}
For $P_{\pi_a}=I_4$ and $P_{\pi_b}$ as in \eqref{eq:Permexb}, we have
$$
\rho=P^{\dagger}_{\pi_a}\rho P_{\pi_b}=P^{\dagger}_{\pi_b}\rho P_{\pi_a}=P^{\dagger}_{\pi_b}\rho P_{\pi_b}.
$$
Therefore, $\rho$  is symmetrically exchangeable (it also satisfies  $\Pi_\mathrm{Sym}\rho\Pi_\mathrm{Sym}=\rho$.) Instead the matrix
\begin{equation}
 \label{eq:densentferm}
\rho=\frac{1}{2}\begin{bmatrix}
 0 & 0 & 0 & 0\\
  0 & 1 & -1 & 0\\
   0 & -1 & 1 & 0\\
    0 & 0 & 0 & 0\\
\end{bmatrix}
\end{equation}
is antisymmetrically exchangeable. It  satisfies $\Pi_\mathrm{Anti}\rho\Pi_\mathrm{Anti}=\rho$
as well as $\rho=-0.5(P^{\dagger}_{\pi_a}\rho P_{\pi_b}+P^{\dagger}_{\pi_b}\rho P_{\pi_a})=P^{\dagger}_{\pi_b}\rho P_{\pi_b}$.
\end{example}

\section{Updating}\label{sec:updat}
Let us assume we measure a subset of particles $x_1,\dots,x_{\check{m}}$ with $\check{m}\leq m$.
Quantum projection measurements can then be described by a collection  
$\{\Pi_i\}_{i=1}^{n\check{m}}$, with $\Pi_i \in \mathcal{H}^{n\check{m} \times n\check{m}}$, of
projection operators that satisfy the completeness equation $\sum_{i=1}^{n\check{m}} \Pi_i =I$. 

We recall the following definition from \cite[Sec.\;S.9.1]{Benavoli2019bb}.
\begin{definition}
 Let  $\bdomain$ be an A-coherent $\star$-exchangeable coherent set of desirable gambles, the set obtained as
\begin{equation}
\label{eq:condition}
\bdomain_{\Pi_i}=\left\{z^{\dagger}Gz  \mid  z^{\dagger}(\Pi_i\otimes I_{m-\check{m}})^{\dagger} G (\Pi_i\otimes I_{m-\check{m}})z \in \bdomain \right\}
\end{equation} 
is  called the {\bf set of desirable gambles conditional} on  $\Pi_i$.
\end{definition} 
We already know \citep{benavoli2016quantum} that updating preserves coherence. We now see that it also preserves exchangeability.

\begin{proposition}
Let  $\bdomain$ be an A-coherent $\star$-exchangeable coherent set of desirable gambles.
 Then $\bdomain_{\Pi_i}$ is an A-coherent $\star$-exchangeable coherent set of desirable gambles 
 on the variables $x_{\check{m}+1},\dots,x_{m}$ and its dual is
 \begin{equation}
\label{eq:rhobayes}
 \mathcal{M}_{\Pi_i}=\left\{\dfrac{(\Pi_i\otimes I_{m-\check{m}})^{\dagger} \rho (\Pi_i\otimes I_{m-\check{m}})}{Tr((\Pi_i\otimes I_{m-\check{m}})^{\dagger} \rho (\Pi_i\otimes I_{m-\check{m}}))} \Big|  \rho \in 
 \mathcal{M}\right\},
\end{equation} 
where $ \mathcal{M}$ is the dual of $\bdomain$.
\end{proposition}
\begin{proof}
 In \cite[Sec.\;S.9.1]{Benavoli2019bb} we have already proved that $\bdomain_{\Pi_i}$ is coherent and that $ \mathcal{M}_{\Pi_i}$
 is the dual of $\bdomain_{\Pi_i}$. Therefore, we only need to prove that $\bdomain_{\Pi_i}$ is a  $\star$-exchangeable coherent set of desirable gambles  on the variables $x_{\check{m}+1},\dots,x_{m}$.
 This means we need to prove that  $z^{\dagger}Gz-\Big(\tfrac{\delta^\star_{l,r}}{2}z^{\dagger}(I_{\check{m}}\otimes P^{m-\check{m}}_{\pi_l})^{\dagger}G(I_{\check{m}}\otimes P^{m-\check{m}}_{\pi_r})z+\tfrac{\delta^\star_{l,r}}{2}z^{\dagger}(I_{\check{m}}\otimes P^{m-\check{m}}_{\pi_l})^{\dagger}G(I_{\check{m}}\otimes P^{m-\check{m}}_{\pi_r})z\Big) \in \bdomain_{\Pi_i}$ for each gamble $z^{\dagger}Gz$.
 This gamble is in $\bdomain_{\Pi_i}$ provided that:
 \begin{align}
 \nonumber
  &z^{\dagger}(\Pi_i\otimes I_{m-\check{m}})^{\dagger}\Big[G-\tfrac{\delta^\star_{l,r}}{2}(I_{\check{m}}\otimes P^{m-\check{m}}_{\pi_l})^{\dagger} G (I_{\check{m}}\otimes P^{m-\check{m}}_{\pi_r})\\
   \nonumber
  &- \tfrac{\delta^\star_{l,r}}{2}(I_{\check{m}}\otimes P^{m-\check{m}}_{\pi_r})^{\dagger} G (I_{\check{m}}\otimes P^{m-\check{m}}_{\pi_l})\Big](\Pi_i\otimes I_{m-\check{m}})z,
 \end{align}
 is in $\bdomain$. By exploiting the following property of the tensor product
  $$
({I_{2}} \otimes {B} )({A} \otimes {I_{1}} )=({A} \otimes {I_{1}} )({I_{2}} \otimes {B} ),
 $$
we need to verify that
 \begin{align}
 \nonumber
  &z^{\dagger}\Big[(\Pi_i\otimes I_{m-\check{m}})^{\dagger}G(\Pi_i\otimes I_{m-\check{m}})\\
   \nonumber
  &-\tfrac{\delta^\star_{l,r}}{2}(I_{\check{m}}\otimes P^{m-\check{m}}_{\pi_l})^{\dagger}(\Pi_i\otimes I_{m-\check{m}})^{\dagger} G (\Pi_i\otimes I_{m-\check{m}})(I_{\check{m}}\otimes P^{m-\check{m}}_{\pi_r})\\
   \nonumber
  &- \tfrac{\delta^\star_{l,r}}{2}(I_{\check{m}}\otimes P^{m-\check{m}}_{\pi_r})^{\dagger} (\Pi_i\otimes I_{m-\check{m}})^{\dagger}G (\Pi_i\otimes I_{m-\check{m}})(I_{\check{m}}\otimes P^{m-\check{m}}_{\pi_l})\Big]z,
 \end{align}
 is in $\bdomain$. This is true because $ \bdomain$ is $\star$-exchangeable.
\end{proof}

\section{Entanglement}\label{sec:ent}
Unlike systems  consisting of distinguishable\footnote{Spatially well-separated indistinguishable particles can be distinguished.} particles, in the case of identical particles the notion  of  entanglement  is  still  under debate
(see e.g. \citep{benatti2014entanglement}). The reason being that, for instance, the two matrices in Example \ref{eq:3} are entangled density matrices for distinguishable particles  and, at the same time, they also satisfy the symmetry and anti-symmetry relationship of identical particles. Are those density matrices entangled in the (anti-)symmetric case?

For distinguishable particles, our gambling formulation of QM allows us to formulate and detect entangled density matrices thorough a Dutch book (sure loss)  \citep{Benavoli2019bb}.
This goes as follows. Given a density matrix $\tilde{\rho}$, we can first compute its dual (an A-coherent  set of desirable gambles):
$$
\bdomain:=\{g(z,z)=z^{\dagger}Gz: L(g)=Tr(G\tilde{\rho})\geq0\}
$$
and then consider its ``classical'' extension
$$
\domain := \text{posi}(\bdomain \cup \mathcal{L}_R^{\geq}).
$$
Hence, $\domain$ is coherent
(under the standard axioms of desirability) provided that $\domain \cap \mathcal{L}_R^{<}=\emptyset$. 

As done in \citep[Sec.4.4]{Benavoli2019bb}, we thus state the following definition. 
\begin{definition}\label{def:ent}
 Let $\tilde{\rho}$ be a  density matrix. Then $\tilde{\rho}$ is entangled if  $\domain \cap \mathcal{L}_R^{<} \neq \emptyset$ ($\domain$ does not avoid sure loss).
\end{definition}

If $\tilde{\rho}$ is not entangled, the bounded linear functionals in its dual can be written as an integral with respect to a probability measure \citep{Benavoli2019bb}:\footnote{To do that, we need to perform another natural extension to the space of all gambles $\text{posi}(\domain \cup \mathcal{L}^{\geq})$.}
 \begin{equation}
  \label{eq:ccc}
   \rho=\int_{ \Omega} zz^{\dagger} d\mu(z).
 \end{equation}
 
As an immediate consequence of Definition \ref{def:ent} and Equation \eqref{eq:ccc} we get:
\begin{proposition}
 Let $\tilde{\rho}$ be a  density matrix, then $\tilde{\rho}$ is not entangled iff it is  a truncated moment matrix (with respect to a standard probability measure $\mu(z)$).
\end{proposition}

The question is therefore how we can extend this result to the case of indistinguishable particles.
In this aim, we need to consider a constraint: \textit{not all Dutch books can be constructed}. In a system of indistinguishable particles, \textit{physical observables} (that is, gambles which can be evaluated through an experiment or, equivalently, physically realisable gambles) must be invariant under all permutations of the $m$ identical particles \citep[XIV.C-4-a]{cohen2020quantum}:
 \begin{equation}
  \label{eq:physical}
g(z,z)=z^{\dagger}Gz = z^{\dagger}\Pi_{\star}G\Pi_{\star}z ~~\forall~z.
 \end{equation}
Based on this constraint, we can thus obtain the following result.
\begin{proposition}
\label{prop:entaindi}
 Let $\tilde{\rho}$ be an entangled $\star$-exchangeable density matrix,  then  the following two statements are equivalent:
 \begin{itemize}
  \item there exists a physical observable  $g(z,z)$ which belongs to 
  $\mathcal{L}_R^{<}$  such that $Tr(G \tilde{\rho})\geq0$;
  \item  $\tilde{\rho}$ cannot be written as
  $$
  \int_{ \Omega} \frac{\Pi_{\star}zz^{\dagger}\Pi_{\star}^{\dagger}}{Tr(\Pi_{\star}zz^{\dagger}\Pi_{\star}^{\dagger})} d\mu(z).
  $$
  for any probability measure $\mu(z)$.
 \end{itemize}
\end{proposition}
\begin{proof}
The results follow by \citep[Sec.4.4]{Benavoli2019bb}. We only need to observe that if $g(z,z)<0$, then
for all probability measures $\mu$:
$$
\begin{aligned}
 0&> \int_{ \Omega} z^{\dagger}Gz d\mu(z)\\
 &=\int_{ \Omega} z^{\dagger}\Pi_{\star}G\Pi_{\star}z d\mu(z)\\
  &=\int_{ \Omega} Tr(G\Pi_{\star}zz^{\dagger}\Pi_{\star}) d\mu(z)\\
&= Tr\left(G\int_{ \Omega}\Pi_{\star}zz^{\dagger}\Pi_{\star} d\mu(z)\right),
\end{aligned}
$$
the second equality follows by the assumption that $g$ is a \textit{physical observable} and thus Equation \eqref{eq:physical}. The above inequality implies that
$$
Tr\left(G\int_{ \Omega}\frac{\Pi_{\star}zz^{\dagger}\Pi_{\star}^{\dagger}}{Tr(\Pi_{\star}zz^{\dagger}\Pi_{\star}^{\dagger})} d\mu(z)\right)<0.
$$
Let $\sigma_{\mu}$ be the result  of the integral (a density matrix). Since the above inequality must hold for all $\mu$, we can rewrite this condition as $\sup_{\mu} Tr(G \sigma_{\mu})<0$. Therefore, any density matrix $\tilde{\rho}$ such that $ Tr(G \tilde{\rho})\geq 0$ must be entangled: it cannot be expressed as an expectation with respect to a probability measure $\mu$.
\end{proof}
 The above result means that particles are entangled when classical probability coherence and quantum A-coherence disagree.
 The first statement tells us that we can use a Dutch book (sure loss) to detect entanglement, but only if the  Dutch book is a physical observable.
Notice that Proposition \ref{prop:entaindi} is in agreement with definitions of entanglement, and ways to detect it, as discussed for instance in \citep{iemini2013quantifying,iemini2013computable,reusch2015entanglement} (in particular see \citep[Eq.\;12]{reusch2015entanglement}).
 
 \begin{example}
 We apply Proposition \ref{prop:entaindi} to the previous  two  particles Example \ref{eq:3}.

\underline{Fermions:} consider the atomic charge (Dirac's delta) $\mu=\delta_{\tilde{z}}
$ with $\tilde{z}=[1,0]^T\otimes[0,1]^T= [0,1,0,0]^T$.
Note that,
$$
\int_{ \Omega} \frac{\Pi_\mathrm{Anti}zz^{\dagger}\Pi_\mathrm{Anti}^{\dagger}}{Tr(\Pi_\mathrm{Anti}zz^{\dagger}\Pi_\mathrm{Anti}^{\dagger})} \delta_{\tilde{z}}(z).
=\frac{1}{2}\begin{bmatrix}
0 & 0 & 0 & 0 \\
0 & 1 & -1 & 0 \\
0 & -1 & 1 & 0 \\
0 & 0 & 0 & 0 \\                                                                                    \end{bmatrix}
$$
which is the density matrix in \eqref{eq:densentferm}. From Proposition \ref{prop:entaindi}, we conclude that that the density matrix is not entangled.

\underline{Bosons:} consider the atomic charge (Dirac's delta) $\delta_{\tilde{z}}
$ with $\tilde{z}=\frac{1}{2}[-\iota,1]^T\otimes[\iota,1]^T= [0.5,-0.5\iota,0.5\iota,0.5]^T$, where $\iota$ is the complex unit.
Note that,
$$
\int_{ \Omega} \frac{\Pi_\mathrm{Sym}zz^{\dagger}\Pi_\mathrm{Sym}^{\dagger}}{Tr(\Pi_\mathrm{Sym}zz^{\dagger}\Pi_\mathrm{Sym}^{\dagger})} \delta_{\tilde{z}}(z).
=\frac{1}{2}\begin{bmatrix}
1 & 0 & 0 & 1 \\
0 & 0 & 0 & 0 \\
0 & 0 & 0 & 0 \\
1 & 0 & 0 & 1 \\                                                                                    \end{bmatrix}
$$
which is the density matrix in \eqref{eq:densentbos}. From Proposition \ref{prop:entaindi}, we conclude that that the density matrix is not entangled.
\end{example}

\section{Conclusions}
In this paper we showed that we can derive the \emph{symmetrization postulate} for indistinguishable particles in QM using the framework of exchangeable desirable gambles. 
Therefore, once again, we proved that QM is a theory of probability: it can be derived by the principles of coherence and computation plus structural assessments of exchangeability. Moreover, we showed that, also in the case of indistinguishable particles, entanglement can be defined (and detected) as a Dutch book: the clash between the QM notion of rationality (which accounts for the principle of computation) and the classical notion of rationality (which assumes infinite computational resources).

We obtained these results by exploiting symmetrization procedures to model structural assessments of indistinguishability. This approach, which is called ``first quantization'' in QM, has a main drawback: it includes  redundant information. More specifically, it potentially allows us to gamble on the state of a single particle which is not a physical observable (it is impossible  in the first place to tell which particle is which). This constitutes a well-known limit in QM.  As an example, we had to impose the condition on physical observables given by Equation \eqref{eq:physical} in order to obtain Proposition \ref{prop:entaindi}.
 
 In QM, there is another formalism to work with indistinguishable particles, called \textit{second quantization}.
Its language allows one to ask the following question ``How many particles are there in each state?''. Since this  formalism does not refer to the labelling of particles, it contains no redundant information.
 As future work, we plan to provide a gambling formulation of QM for the \textit{second quantization},  exploring the connection with the  count vectors formalism developed by  \cite{de2012exchangeability}.

\section*{Author Contributions}
All authors have contributed equally to the manuscript.

\bibliography{biblio}
 
\end{document}

%% file: macros.tex
\usepackage{multirow}
\usepackage{bm}
\usepackage{amsmath, amssymb, amsfonts}
\usepackage{array}
\usepackage{enumerate}
\usepackage{bbm}
\usepackage{color}
\usepackage{float}
\usepackage[mathscr]{euscript}
\usepackage{mathtools}
\usepackage{framed}
\usepackage{setspace}
\usepackage{mdframed}

\usepackage{thmtools}
\usepackage{environ}
\usepackage[draft]{todonotes} 
  \usepackage{times}
\usepackage{pgfplots}
\usetikzlibrary{matrix}
\usepackage[toc,page]{appendix}

  \usetikzlibrary{positioning,calc,fit,shapes.geometric,patterns}
  \pgfdeclarelayer{background}
  \pgfdeclarelayer{foreground}
  \pgfsetlayers{background,main,foreground}
  \tikzstyle{vec}=[circle,inner sep=1pt,outer sep=-1pt,fill]
  \tikzstyle{border}=[thick]
  \tikzstyle{favborder}=[border,dotted]
  \tikzstyle{exclborder}=[border,dashed]
\usepackage{etex,etoolbox}
\usepackage{hyperref}
 \usepackage{calrsfs}

  \usetikzlibrary{positioning,calc,fit,shapes.geometric,patterns}
  \pgfdeclarelayer{background}
  \pgfdeclarelayer{foreground}
  \pgfsetlayers{background,main,foreground}
  \tikzstyle{vec}=[circle,inner sep=1pt,outer sep=-1pt,fill]
  \tikzstyle{border}=[thick]
  \tikzstyle{favborder}=[border,dotted]
  \tikzstyle{exclborder}=[border,dashed]
\usepackage{etex,etoolbox}
\usepackage{hyperref}

\hypersetup{colorlinks=true,linkcolor=blue,citecolor=magenta}

\usepackage{etex,etoolbox}
\usepackage{hyperref}

\newenvironment{tbs}{%
   \small\tt
   \begin{enumerate}[$\blacktriangleright$]}{\end{enumerate}}
\newcommand{\btbs}{\begin{tbs}}                                                                      
\newcommand{\etbs}{\end{tbs}}

\newcommand{\Reals}{\mathbb{R}} 

\newcommand{\complex}{\mathbb{C}}

 \DeclareMathOperator{\posi}{posi}

\newcommand{\pspace}{\Omega}

\newcommand{\reals}{\Reals}

\newcommand{\gambles}{\mathcal{L}}

\newcommand{\He}{\mathcal{H}}

\newcommand{\domain}{\mathcal{K}}

\newcommand{\sign}{\mathrm{sign}}

\newcommand{\bdomain}{\mathcal{C}}

\newcommand{\nonnegative}{\gambles^{\geq} }
\newcommand{\bnonnegative}{\Sigma^{\geq}}

\newcommand{\negative}{\gambles^{<} }
